\documentclass[a4paper,11pt]{article}
\usepackage[utf8x]{inputenc}
\usepackage{graphicx}
\usepackage[english]{babel}
\usepackage[T1]{fontenc}
\usepackage{amsmath}
\usepackage{amsfonts}
\usepackage{amssymb}
\usepackage{amsthm}
\usepackage{amstext}
\usepackage{pstricks}
\usepackage{color}
\usepackage{dsfont}
\usepackage{cite}
\usepackage{tikz}
\usetikzlibrary{matrix,arrows,decorations.pathmorphing}

\newtheorem{theorem}{Theorem}[section]
\newtheorem{lemma}[theorem]{Lemma}

\newtheorem{proposition}[theorem]{Proposition}

\theoremstyle{definition}

\theoremstyle{remark}

\def \cA {\mathcal{A}}

\def \cG {\mathcal{G}}

\def \cQ {\mathcal{Q}}

\def \a {\alpha}
\def \b {\beta}

\def \d {\delta}
\def \e {\varepsilon}

\def \k {\kappa}
\def \l {\lambda}

\def \D {\Delta}

\def \N {\mathbb{N}}

\def \R {\mathbb{R}}

\def \lra {\longrightarrow}

\def \al {\mathcal{A}^\ell}

\def \zl {\lbrace\,0,\dots,\ell\,\rbrace}

\def \exa {e^{-a}}

\def \hyperc {\lbrace\,0,1\,\rbrace^\ell}

\begin{document}

\begin{center}
\begin{Large}
Asymptotic behavior of Eigen's quasispecies model
\end{Large}

\begin{large}
Joseba Dalmau

\vspace{-12pt}
École Polytechnique

\vspace{4pt}
\today
\end{large}

\begin{abstract}
\noindent
We study Eigen's quasispecies model
in the asymptotic regime 
where the length of the genotypes goes to $\infty$
and the mutation probability goes to 0.
A limiting infinite system of differential equations is obtained.
We prove the convergence of the trajectories,
as well as the convergence of the 
equilibrium solutions.
We give the analogous results for a discrete--time version
of Eigen's model, which coincides with a model 
proposed by Moran.
\end{abstract}
\end{center}

\section{Introduction}\label{Intro}
In the early 70s, 
Manfred Eigen proposed a mathematical model 
for the evolution of a prebiotic population,
under the complementary forces of selection
and mutation~\cite{Eig}.
Let $\cG$ be the finite set of possible genotypes,
along with a fitness function $f:\cG\lra\R^+$
and a mutation matrix $(M(u,v),u,v\in\cG)$.
The concentration $x_v(t)$ of individuals
having genotype $v$ in the population,
evolves according to the differential equation
$$x_v'(t)\,=\,\sum_{u\in\cG}x_u(t)f(u)M(u,v)
-x_v(t)\sum_{u\in\cG}x_u(t)f(u)\,.$$
Under the assumption that the matrix
$W=(f(u)M(u,v),u,v\in\cG)$ is primitive,
it is well known~\cite{JER76,TM74,SS82} that
the above system of differential equations
has a unique stationary solution $x^*$
and all the trajectories converge to $x^*$.
Eigen's model exhibits two phenomena of particular importance:
an error threshold phenomenon and a quasispecies distribution.
In order to see what this means, 
let us fix $\cG$ to be the $\ell$--dimensional hypercube $\hyperc$,
and suppose that mutations arrive 
independently on each site of a chain with probability $q$.
When the length of the genotypes $\ell$ tends to infinity,
there exists a critical mutation probability $q^*$,
called the error threshold,
separating two different regimes.
For mutation probabilities above the error threshold,
the population at equilibrium is totally random.
For mutation probabilities below the error threshold,
the population at equilibrium possesses a positive 
concentration of the fittest genotype,
along with a cloud of mutants which are a few mutations away from
the fittest genotype.
This kind of distribution is referred to as a quasispecies.
Explicit formulas for the distribution of the quasispecies
have been found in~\cite{CDClassdep},
in the case of the sharp peak landscape,
as well as in the case of class--dependent fitness landscapes.
In order to obtain these formulas, the following asymptotic regime
is considered:
$$\ell\to\infty\qquad
q\to 0\qquad 
\ell q\to a\in \,]0,+\infty[\,\,.$$
In this asymptotic regime, 
an infinite version of Eigen's system of differential equations
$(x^\infty)'(t)=F(x^\infty(t))$ 
is obtained,
and the equations for the stationary solutions of 
the infinite system are explicitly solved.
The solutions $\cQ(f,a)$ are called quasispecies distributions,
and depend both on the fitness function $f$  and the mean number of
mutations per genome per generation $a$.
The aim of this paper is to complete the picture 
by showing the following convergences:
$$
\begin{tikzpicture}[scale=1.5]
\node (A) at (0,2) {$x(t)$};
\node (B) at (2,2) {$x^*$};
\node (C) at (0,0) {$x^\infty(t)$};
\node (D) at (2,0) {$\cQ(f,a)$};
\path[->,font=\normalsize,>=angle 90]
(A) edge node[above]{$t\to\infty$} (B)
(A) edge node[left]{$\genfrac{}{}{0pt}{1}{\ell\to\infty,\,
q\to 0}
{{\ell q} \to a}$} (C)
(B) edge node[right]{$\genfrac{}{}{0pt}{1}{\ell\to\infty,\,
q\to 0}
{{\ell q} \to a}$} (D)
(C) edge node[above]{$t\to\infty$} (D);
\end{tikzpicture}$$
We also obtain analogous results for a closely related model.
We keep the same framework as in Eigen's model,
but we consider the time to be discrete.
The concentration $x_v(n)$ of individuals having genotype $v$
evolves according to the following dynamical system:
$$x_v(n+1)\,=\,\frac{\displaystyle\sum_{u\in\cG}x_u(n)f(u)M(u,v)}
{\displaystyle\sum_{u\in\cG}x_u(n)f(u)}\,.$$
This model was first proposed by Moran~\cite{Moran76},
and it is not to be confused with the well--known Moran model,
which is a stochastic model for the evolution of a finite population.
We thus call this model the deterministic Moran model.
We will show that a similar diagram holds for the deterministic Moran model.
In particular, the unique fixed point of the above dynamical system 
is again $x^*$, and the distribution of the quasispecies found in the limit
coincides with the distribution of the quasispecies for Eigen's model.

\section{Models and known results}\label{Modkres}
Let $\cA$ be a finite alphabet of cardinality $\k$, and let
$\cA^\ell$ be the set of sequences of length $\ell\geq1$ over $\cA$.
We will refer to $\al$ as the set of genotypes;
typical choices for $\cA$ are $\lbrace\,A,T,G,C\,\rbrace$
for DNA sequences, the set of the twenty amino acids for proteins,
or $\lbrace\,0,1\,\rbrace$ for binary sequences.
The models we consider aim at modeling the evolution
of a population, the individuals in the population having genotypes in $\al$.
The evolution will be guided by two main forces,
selection and mutation.
Selection is defined via a fitness function, that is,
a mapping $g:\al\lra\,]0,+\infty[$.
Mutations arrive independently on each site of the genotype,
with probability $q\in(0,1)$; when a mutation occurs in a certain site, 
the letter present in it is replaced by one of the $\k-1$ remaining letters in the 
alphabet, chosen uniformly at random.
The natural distance on $\al$ is the Hamming distance,
which counts the number of different digits between two chains, i.e.,
$$\forall\,u,v\in\al\qquad
d(u,v)\,=\,\text{card}\big\lbrace\,1\leq i\leq\ell:
u(i)\neq v(i)\,\big\rbrace\,.$$
The probability that a genotype $u$ is transformed into a genotype $v$
by mutation is given by
$$Q(u,v)\,=\,\Big(
\frac{q}{\k-1}
\Big)^{d(u,v)}(1-q)^{\ell-d(u,v)}\,.$$
We make the following two assumptions on  the fitness function $g$.

\textbf{Assumptions.} We suppose that: \\
(A1) There exists a privileged sequence, $w^*\in\al$
of strictly maximal fitness, that is, $g(w^*)>g(u)$ 
for all $u\in\al\setminus\lbrace\,w^*\,\rbrace$.\\
(A2) All sequences at a same distance from $w^*$ share the same fitness.

The privileged sequence $w^*$ will be referred to as the master sequence.
Under these two assumptions, we can decompose the space of genotypes $\al$
into Hamming classes with respect to the master sequence.
We say that a genotype  $u\in\al$ belongs to the Hamming class $k$
if $d(u,w^*)=k$. The set of Hamming classes is $\zl$ and, under assumption
(A2), we can define a fitness function $f:\zl\lra\,]0,+\infty[\,$ by setting 
$f(k)$ to be the value of $g(u)$ common to all the sequences $u$ in the Hamming class $k$.
Moreover, the mutation matrix $Q$ can be factorized through the Hamming classes.
Indeed, the probability that a genotype in the class $i$
mutates into a genotype in the class $k$ is given by
$$M(i,k)\,=\,P\big(
i-Bin(i,q/(\k-1))+Bin(\ell-i,q)\,=\,k
\big)\,,$$
where $Bin(n,p)$ is the binomial law,
and the two binomials in the formula are independent.
Define $\D^\ell$ to be the $\ell$--dimensional unit simplex:
$$\D^\ell\,=\,\big\lbrace\,
x\in[0,1]^{\ell+1}:x_0+\cdots+x_\ell=1
\,\big\rbrace\,.$$
\textbf{Eigen's model.} Let $x_k(t)$ represent
the proportion of individuals in the class $k$
in a population at time $t$.
The quantities $x_k(t)$ evolve according to the following
system of differential equations:
$$(Eig)\qquad
x_k'(t)\,=\,\sum_{i=0}^\ell x_i(t)f(i)M(i,k)
-x_k(t)\sum_{i=0}^\ell x_i(t)f(i)\,,\qquad 0\leq k\leq \ell\,.$$
Note that if $x^0$ belongs to $\D^\ell$,
then the solution $(x(t),t\geq0)$ of $(Eig)$
with initial condition $x(0)=x^0$ belongs to $\D^\ell$
for all $t\geq0$, which is a direct consequence of $M$ 
being a stochastic matrix.

\textbf{The deterministic Moran model.} Let $x_k(n)$
represent the proportion of individuals in the class $k$
in generation $n$.
The quantities $x_k(n)$ evolve according to the following
discrete--time dynamical system:
$$(DM)\qquad
x_k(n+1)\,=\,\frac{\displaystyle \sum_{0\leq i\leq\ell} x_i(n)f(i)M(i,k)}
{\displaystyle\sum_{0\leq i\leq\ell} x_i(n)f(i)}\,,\qquad 0\leq k\leq\ell\,.$$
Again, if $x^0\in\D^\ell$, then the solution $(x(n),n\geq0)$
of $(DM)$ with initial condition $x(0)=x^0$
belongs to $\D^\ell$ for all $n\geq0$.

Let us define the matrix $W$ by
$$\forall\, i,j\in\zl\qquad
W(i,j)\,=\,f(i)M(i,j)\,.$$ 
%
By assumption (A1), the matrix $W$ is strictly positive,
and thus the Perron--Frobenius theorem applies.
We have the following result.

\begin{proposition}\label{sseigen}
Eigen's system of differential equations $(Eig)$ admits a unique 
stationary solution $x^*\in\D^\ell$. Moreover,
for every $x^0\in\D^\ell$, the solution $(x(t),t\geq0)$ of $(Eig)$
with initial condition $x(0)=x^0$ satisfies
$$\lim_{t\to \infty} x(t)\,=\,x^*\,.$$
\end{proposition}
This result is well--known,
and has been established by several authors,
see for instance~\cite{BK83,Jones77,JER76,TM74}.
A similar result holds for the deterministic Moran model,
which has been proven by Moran himself in~\cite{Moran76}.
Both results can be proven in a similar way,
by using the Perron--Frobenius theorem.
In fact, the vector $x^*$ is the same in both cases,
and it is the left Perron--Frobenius eigenvector of the matrix $W$,
normalized so that it belongs to $\D^\ell$.
We also remark that the mean fitness of the population at equilibrium,
$$\l\,=\,\sum_{0\leq i\leq\ell}x_i^*f(i)\,,$$
is the Perron--Frobenius eigenvalue of the matrix $W$.

For an error threshold phenomenon to take place,
we consider the regime where the length of the genomes goes to infinity.
More explicitly, we consider the asymptotic regime
$$\ell\to\infty\,,\qquad
q\to 0\,,\qquad 
\ell q\to a\in\, ]0,+\infty[\,.$$
Recall that the mutation matrix $M$ 
is given by the formula 
$$M(i,k)\,=\,P\big(
i-Bin(i,q/(\k-1))+Bin(\ell-i,q)\,=\,k
\big)\,.$$
When considering the above asymptotic regime,
the first of the binomial laws converges to a Dirac mass at 0,
while the second one converges to a Poisson distribution
of parameter $a$. Therefore, we obtain an infinite
mutation matrix $M_\infty$, which is given by
$$\forall\,i,k\geq0\qquad
M_\infty(i,k)\,=\,\begin{cases}
\quad \displaystyle \exa\frac{a^{k-i}}{(k-i)!}\quad
&\text{if}\quad k\geq i\,,\\
\quad 0\quad
&\text{if}\quad k<i\,.
\end{cases}$$
We suppose that there exists a function $f_\infty:\N\lra\,]0,+\infty[\,$,
such that for each $\ell$, the restriction of $f_\infty$
to $\zl$ is equal to the fitness function $f$.
We suppose that the fitness function $f_\infty$
satisfies the following assumption.

\textbf{Assumption (B).} The fitness function $f_\infty$ is positive, has a strict maximum at $0$, and converges to 1, i.e.,
$$\forall\,k\geq1\quad
f_\infty(0)>f_\infty(k)>0\qquad 
\text{and}\qquad\lim_{k\to\infty}f_\infty(k)=1\,.$$
We consider the following limiting systems.

\textbf{Eigen's infinite system.} Let $y_k(t)$ 
represent the proportion of individuals in the class $k$ 
in a population at time $t$.
The quantities $y_k(t)$ evolve according to the following
system of differential equations:
$$(Eig_\infty)\quad
y_k'(t)\,=\,\sum_{i=0}^k
y_i(t)f_\infty(i)\exa \frac{a^{k-i}}{(k-i)!}-
y_k(t)\sum_{i\geq 0}y_i(t)f_\infty(i)\,,\quad
k\geq0\,.$$
\textbf{The infinite deterministic Moran model.}
Let $y_k(n)$ represent the proportion of individuals
in the class $k$ in generation $n$.
The quantities $y_k(n)$
evolve according to the following discrete--time dynamical system:
$$(DM_\infty)\qquad
y_k(n+1)\,=\,\frac{\displaystyle\sum_{0\leq i\leq k} y_i(n)f_\infty(i)\exa\frac{a^{k-i}}{(k-i)!}}{\displaystyle\sum_{i\geq0}y_i(n)f_\infty(i)}\,,\qquad k\geq 0\,.
$$
We first look for the stationary solutions of $(Eig_\infty)$,
which coincide with the fixed points of $(DM_\infty)$.
We restrict our attention to the stationary solutions 
satisfying 
$$\sum_{k\geq0} y_k\,=\,1\,.$$
Let $I(f_\infty)\in\N$ be the set of indices $i$ such that $$f_\infty(i)\exa>1\qquad \text{and}\qquad 
f_\infty(i)>f_\infty(j)\quad \forall\,j>i\,.$$
Under assumption (B), we have the following result.
\begin{proposition}
The system $(Eig_\infty)$ has as many stationary solutions 
as there are elements in $I(f_\infty)$. Moreover, 
for each $i\in I(f_\infty)$, the associated solution $(\rho^i_k)_{k\geq0}$
satisfies
$$\rho^i_0=\cdots=\rho^i_{i-1}=0\qquad
\text{and}\qquad\rho^i_i>0\,.$$
\end{proposition}
A similar statement holds for the fixed points of $(DM_\infty)$.
This result has been proven in~\cite{CDClassdep},
where an explicit formula is found for the solutions $\rho^i$.
Indeed, the solution $\rho^i$ is given by: 
for all $k\geq 0$,
$$\rho^i_{i+k}\,=\,\frac{\displaystyle
\frac{1}{f(i)}1_{k=0}+
\frac{a^k}{f(i+k)}\!\!\!\!\!\!\!\!
\sum_{\genfrac{}{}{0pt}{1}{1\leq h\leq k}{0=i_0<\cdots<i_h=k}}
\prod_{t=1}^h \frac{f(i+i_t)}{(i_t-i_{t-1})!(f(i)-f(i+i_t))}}
{\displaystyle \frac{1}{f(i)}+\!\!\!\!\!\!\sum_{\genfrac{}{}{0pt}{1}{h\geq 1}{0=i_0<\cdots<i_h}}
\frac{a^{i_h}}{f(i+i_h)}\prod_{t=1}^h\frac{f(i+i_t)}{(i_t-i_{t-1})!(f(i)-f(i+i_t))}}\,,$$
where an empty sum is taken to be equal to 0,
and the index $\infty$ has been omitted from the fitness function. From now on, we will always omit
the index $\infty$ in the fitness function, 
and we will denote by $f$ both the fitness
function on $\N$ and its restriction to $\zl$.

Before stating our results, we justify the existence and uniqueness of
a global solution of the system $(Eig_\infty)$ for a given initial condition.
The facts stated below follow from the general theory of ODE's on Banach spaces (see for instance~\cite{Cartan}, part II, chapter 1).
We denote by $\ell^1$ the space of absolutely summable sequences
$(y_k)_{k\geq0}$ and by $||\cdot||$ their $\ell^1$ norm, as well as 
the operator norm associated to it.
Define the operator $W_\infty:\ell^1\lra\ell^1$ by setting
$$\forall\,y\in\ell^1\quad
\forall\,k\geq0\qquad
(yW_\infty)_k\,=\,\sum_{i=0}^k y_i f(i)M_\infty(i,k)\,.$$
In view of assumption (B), the operator $W_\infty$ is bounded
by $f(0)$. Moreover,
we can rewrite the system of differential equations
$(Eig_\infty)$
in terms of the operator $W_\infty$ as $y'(t)=F(y(t))$ with
$$F(y)\,=\,yW_\infty-y\langle yW_\infty,1\rangle\,,$$
where for $y\in\ell^1$ and $(h_k)_{k\geq0}$ 
a bounded sequence, $\langle y,h\rangle=\sum_{i\geq0}y_ih_i$.
Since the operator $W_\infty$ is bounded,
the mapping $F:\ell^1\lra\ell^1$ is locally Lipschitz.
Indeed, let $y\in\ell^1$ and $\d>0$,
for every $z\in\ell^1$ such that $||y-z||<\d$, we have
$$||F(y)-F(z)||\,\leq\,||y-z||\cdot||W_\infty||+
||y-z||\cdot\big|\langle yW_\infty,1\rangle\big|+
||z||\cdot\big|\langle (y-z)W_\infty,1\rangle\big|\,.$$
Note that for every $u\in\ell^1$ we have 
$|\langle uW_\infty,1\rangle|\leq ||u||\cdot||W_\infty||$. Thus,
$$||F(y)-F(z)||\,\leq\,||W_\infty||\big(1+||y||+||z||\big)||y-z||\,\leq\,
M(y,\d)||y-z||\,,$$
with $M(y,\d)=||W_\infty||(1+2||y||+\d)$, so that $F$ is locally Lipschitz.
Therefore, the Cauchy problem $y'(t)=F(y(t))$
with initial condition $y(0)\in\ell^1$ admits a unique maximal solution
$y:\,]a,b[\,\lra\ell^1$
with $-\infty\leq a<0<b\leq +\infty$.
Furthermore,
the set 
$$E\,=\,\big\lbrace\,
y\in\ell^1: \forall k\geq0\ y_k\geq0\quad\text{and}\quad ||y||=1
\,\big\rbrace$$
is positively invariant, that is,
if $y(0)\in E$, then for all $t\geq0$ we have $y(t)\in E$.
Indeed, if $y(0)$ is a non--negative sequence, 
the fact that $y_k(t)\geq0$ for all $k\geq0$ and $t\geq0$
follows from lemma~\ref{posilemma} together with an inductive argument.
Moreover, if $y\in\ell^1$ is a non--negative sequence,
$$\frac{d}{dt}||y(t)||\,=\,
\sum_{k\geq0}y'_k(t)\,=\,
\langle yW_\infty,1\rangle
\big(1-||y(t)||\big)\,.$$
Thus, $E$ is positively invariant.
For every $y\in E$ we have $||F(y)||\leq 2||W_\infty||$, therefore
the solution $y$ does not explode and $b$ can be taken to be equal to $\infty$. 
In the sequel, we will only consider solutions
of $(Eig_\infty)$ such that $y(0)\in E$.
In this case, the limit of $y_k(t)$ when $t$ goes
to $\infty$ is well defined for all $k\geq0$.
We now proceed to state our main results.

\section{Main results}\label{Mainres}
We begin by showing the convergence
of the solutions of the system $(Eig_\infty)$.
Assume that $I(f)=\lbrace\,i_1,\dots,i_N\,\rbrace$. Note that $N$ might be equal to 0,
in which case $I(f)$ would be empty.
\begin{theorem}\label{solstatstab}
Let $(y(t))_{t\geq0}$ be a solution of $(Eig_\infty)$.
For every $k\geq0$ and
$h\in\lbrace\,1,\dots,N\,\rbrace$,
$$\lim_{t\to\infty}y_k(t)\,=\,\rho^{i_h}_k$$
if and only if the initial condition satisfies
$$y_0(0)\,=\,\cdots\,=\,y_{i_{h-1}}(0)\,=\,0
\qquad\text{and}\qquad
\max_{i_{h-1}<i\leq i_h}y_i(0)\,>\,0\,.$$
In this case, $y(t)$ converges to $\rho^{i_h}$
in $\ell^1$.
Otherwise, $y_k(t)$ converges to $0$ for 
all $k\geq0$.
\end{theorem}
Next, we show that the solutions of $(Eig)$
converge to the solutions of $(Eig_\infty)$ on finite time intervals.
\begin{theorem}\label{convtraj}
Let $(x(t))_{t\geq0}$ and $(y(t))_{t\geq0}$
be solutions of $(Eig)$ and $(Eig_\infty)$
respectively,
and assume that the initial conditions converge, i.e.,
$$\forall\, k\geq 0\,,\qquad
\lim_{
\genfrac{}{}{0pt}{1}{\ell\to\infty,\,
q\to 0}
{{\ell q} \to a}
}
\,x_k(0)\,=\,y_k(0)\,.
$$
Then, for every $T>0$ 
and for every $k\geq 0$,
$$\lim_{
\genfrac{}{}{0pt}{1}{\ell\to\infty,\,
q\to 0}
{{\ell q} \to a}
}
\,\sup_{0\leq t\leq T}\,
|x_k(t)-y_k(t)|\,=\,0\,.$$
\end{theorem}
Finally, we study that the convergence of 
the unique stationary solution of $(Eig)$.
\begin{theorem}\label{convpstat}
Let $x^*=(x^*_k)_{0\leq k\leq\ell}$
be the unique stationary solution of $(Eig)$.
We have the following dichotomy:

$\bullet$ If $f(0)\exa\leq 1$,
$$\forall\, k\geq 0\,,\qquad \lim_{
\genfrac{}{}{0pt}{1}{\ell\to\infty,\,
q\to 0}
{{\ell q} \to a}
}
\,x^*_k\,=\,0\,.$$

$\bullet$ If $f(0)\exa> 1$,
$$\forall\, k\geq 0\,,\qquad \lim_{
\genfrac{}{}{0pt}{1}{\ell\to\infty,\,
q\to 0}
{{\ell q} \to a}
}
\,x^*_k\,=\,\rho^0_k
\,,$$
where $(\rho^0_k)_{k\geq0}$ 
is the unique stationary solution of $(Eig_\infty)$
satisfying $\rho^0_0>0$.
\end{theorem}
Analogous results hold for the discrete--time models $(DM)$
and $(DM_\infty)$. The proofs are similar in both cases,
and thus, in what follows, we will only deal with the continuous--time
case.
The next three sections prove each of the above results.

\section{Convergence to equilibrium}\label{Convequi}
The aim of this section is to prove the theorem~\ref{solstatstab}.
We will only show that 
if $i_1=0$ and $y_0(0)>0$, then for every
$k\geq 0$ ,
$$\lim_{t\to\infty}y_k(t)\,=\,\rho_k^0\,,$$
where $(\rho^0_k)_{k\geq0}$ 
is the stationary solution of $(Eig_\infty)$
associated to $0$.
The remaining cases can be shown in a similar fashion. We denote by $\phi_\infty(t)$
the mean fitness of the system $(Eig_\infty)$,
i.e., $\phi_\infty(t)=\sum_{i\geq0}y_i(t)f(i)$.
Let us note first that if $y_0(0)>0$,
then $y_0(t)>0$ for all $t\geq 0$.
Indeed, since for all $t\geq0$
we have $\phi_\infty(t)\leq f(0)$,
$$y_0'(t)\,=\,y_0(t)f(0)\exa-y_0(t)\phi_\infty(t)
\,\geq\, y_0(t)f(0)(\exa-1)\,.$$
Therefore, for all $t\geq0$,
$$y_0(t)\,\geq\,
y_0(0)e^{-f(0)(1-\exa)t}\,>\,0\,.$$
We can thus make the following change of variables:
for all ${k\geq0}$,
we set $z_k(t)=y_k(t)/y_0(t)$.
Differentiating, we obtain a new system of differential equations:
$$
z_k'(t)\,=\,\sum_{i=0}^k
z_i(t)f(i)\exa\frac{a^{k-i}}{(k-i)!}
-z_k(t)f(0)\exa\,,\qquad
k\geq1\,.$$
Thanks to this change of variables,
we have managed to transform the original 
system of differential equations into a linear system.
We will show by induction that for all $k\geq 0$,
$z_k(t)$ converges to $z_k^*$ when $t$ goes to infinity,
where $z_k^*=\rho^0_k/\rho^0_0$.
The result is obvious for $k=0$, 
since $z_0(t)=1$ for all $t\geq0$.
Let $k\geq 1$ and suppose that $z_i(t)$
converges to $z_i^*$
for $i\in\lbrace\,0,\dots,k-1\,\rbrace$.
We have
$$z_k'(t)\,=\,
\sum_{i=0}^{k-1}z_i(t)f(i)\exa\frac{a^{k-i}}{(k-i)!}
-(f(0)-f(k))\exa z_k(t)\,.$$
We conclude that (appendix~\ref{LemmesEdo}):
$$
\lim_{t\to\infty}y_k(t)\,=\,
\frac{1}{f(0)-f(k)}\sum_{i=0}^{k-1}
z_i^* f(i)\frac{a^{k-i}}{(k-i)!}\,=\,
z_k^*\,.
$$
This concludes the induction step.
It remains to prove that
$$\lim_{t\to\infty}y_0(t)\,=\,\rho_0^0\,.$$
We have:
$$y_0(t)\,=\,\Bigg(
\sum_{k=0}^\infty
z_k(t)\Bigg)^{-1}\qquad
\text{and}\qquad
\rho_0^0\,=\,\Bigg(
\sum_{k=0}^\infty
z_k^*\Bigg)^{-1}\,.$$
First, we will prove the convergence 
assuming that the fitness function $f$
is eventually constant,
and we will then use this fact to 
prove the general case.
Let us suppose the existence of an $N\geq 0$ 
such that the fitness function $f$
is constant and equal to $1$
for all $n>N$. 
In this case, the mean fitness
$\phi_\infty(t)$ 
is a function of $y_0(t),\dots,y_N(t)$,
$$\phi_\infty(t)\,=\,\sum_{0\leq k\leq N}
y_k(t)(f(k)-1)+1\,=\,
y_0(t)
\sum_{0\leq k\leq N}
z_k(t)(f(k)-1)+1\,.
$$
Likewise, the mean fitness at equilibrium, $\phi_\infty^*$,
is a function of $\rho^0_0,\dots,\rho^0_N$:
$$\phi^*_\infty\,=\,\sum_{0\leq k\leq N}
\rho^0_k(f(k)-1)+1\,=\,
\rho^0_0
\sum_{0\leq k\leq N}
z_k^*(f(k)-1)+1\,.$$
Yet, $\phi^*_\infty=f(0)\exa$.
We conclude that
$$\sum_{0\leq k\leq N}z_k^*(f(k)-1)\,=\,
\frac{f(0)\exa-1}{\rho^0_0}\,>\,0\,.$$
Set
\begin{align*}
\a(t)\,&=\,\sum_{0\leq k\leq N}
z_k(t)(f(k)-1)\,,\\
\a^*\,&=\,\sum_{0\leq k\leq N}
z_k^*(f(k)-1)\,,\\
\b\,&=\,f(0)\exa-1\,.
\end{align*}
The differential equation for $y_0(t)$ 
can be rewritten as
$$y_0'(t)\,=\,y_0(t)\big(
\b-y_0(t)\a(t)
\big)\,.$$
We will show that $y_0(t)$ converges
to $\b/\a^*\,=\,\rho^0_0$.
Let $\e>0$ be small enough so that $\a^*-\e>0$,
and let $T\geq0$ be large enough so that
$$\forall\, t\geq T\qquad
|\a(t)-\a^*|\,<\,\e\,.$$
Then, for all $t\geq T$,
the derivative $y_0'(t)$ is strictly positive over
$]0,\b/(\a^*+\e)[$ 
and strictly negative over $]\b/(\a^*-\e),1]$.
We deduce the existence of a $T_1>T$ such that
$$\forall\, t\geq T_1\qquad
\frac{\b}{\a^*+\e}\,\leq\,y_0(t)\,\leq\,
\frac{\b}{\a^*-\e}\,.$$
Letting $\e$ go to 0,
we obtain the convergence of $y_0(t)$ towards $\rho^0_0$.
In particular, we get the convergence
$$\lim_{t\to\infty}\ \sum_{k\geq0}z_k(t)\,=\,\sum_{k\geq0}z_k^*\,.$$
If $f$ is not eventually constant,
we choose $\e>0$ small enough so that
$f(0)\exa>1+\e$ and $N\geq 0$ large enough so that
$$\forall\,n>N\qquad
f(n)\,<\,1+\e\,.$$
Let $f_N:\N\lra\R^+$ be the mapping defined by:
$$\forall\, n\geq0\qquad
f_N(n)\,=\,\begin{cases}
\quad f(n)&\quad\text{if}\quad n\leq N\,,\\
\quad 1+\e&\quad\text{if}\quad n>N\,.
\end{cases}$$
Consider the system of differential equations
$$u_k'(t)\,=\,
\sum_{0\leq i\leq k}
u_i(t)f_N(i)\exa\frac{a^{k-i}}{(k-i)!}
-u_k(t)f_N(0)\exa\,.$$
Since $f\leq f_N$ and $f(0)=f_N(0)$,
if $y_k(0)\leq u_k(0)$ for all $k\geq0$, we have,
thanks to the lemma~\ref{compedolin},
$$\forall\, k\geq0\quad
\forall\, t\geq0\qquad
y_k(t)\,\leq\,u_k(t)\,.$$
Moreover, since $f_N$ 
is eventually constant,
the series with general term
$u_k(t)$ converges
to the series with general term $u_k^*$.
We conclude that the same holds for
the series with general term $z_k(t)$, 
as wanted.
It remains to see that $y(t)$ converges to $\rho^0$
in $\ell^1$.
Let $\e>0$ and choose $N\geq0$ large enough
so that $\rho^0_0+\cdots+\rho^0_N>1-\e/4$.
It follows from the argument above that
there exists $T>0$ such that
$$\forall\,k\in\lbrace\,0,\dots,N\,\rbrace\quad
\forall\,t\geq T\qquad
|y_k(t)-\rho_k^0|\,<\,\frac{\e}{4(N+1)}\,.$$
In particular, for $t\geq T$,
$$
\sum_{k>N} y_k(t)\,=\,
1-\sum_{k=0}^N y_k(t)
\,\leq\,\bigg|
1-\sum_{k=0}^N \rho^0_k
\bigg|+
\sum_{k=0}^N|\rho^0_k-y_k(t)|\,<\,\frac{\e}{2}\,.
$$
Then, for all $t\geq T$,
$$||y(t)-\rho^0||\,\leq\,
\sum_{k=0}^N |y_k(t)-\rho^0_k|+
\sum_{k>N} y_k(t)+\sum_{k>N}\rho^0_k\,<\,\e\,,
$$
which proves the $\ell^1$ convergence.

\section{Convergence of the trajectories}\label{Convtraj}
The aim of this section is to prove the theorem~\ref{convtraj}.
Let $\e,\d,T>0$ and let $N$ be large enough so that
$$\forall\, n\geq N\qquad
|f(n)-1|\,<\,\d\,.$$
We will show that,
for every $n\geq N$ and $t\leq T$,
asymptotically,
$$
\sum_{k=0}^n |x_k(t)-y_k(t)|\,<\,\e\,.
$$
Let $n\geq N$. Asymptotically,
for every $k\in\lbrace\,1,\dots,n\,\rbrace$,
\begin{align*}
\forall\, i\in\lbrace\,0,\dots,k\,\rbrace\,,&\qquad
\big|M_H(i,k)-M_\infty(i,k)\big|\,<\,\d\,,\\
\forall\, i\in\lbrace\,k+1,\dots,\ell\,\rbrace\,,&\qquad
M_H(i,k)\,<\,\d\,,\\
&\qquad \big|x_k(0)-y_k(0)\big|\,<\,\d\,.
\end{align*}
Moreover, 
denoting by $\phi_H(t)$ and $\phi_\infty(t)$
the mean fitness of the systems $(Eig)$ and $(Eig_\infty)$, for every $t\geq0$,
\begin{align*}
\Bigg|\phi_H(t)-\sum_{k=0}^N x_k(t)(f(k)-1)-1\Bigg|\,&=\,
\Bigg|\sum_{k=N+1}^\ell x_k(t)(f(k)-1)\Bigg|\,<\,\d\,,\\
\Bigg|\phi_\infty(t)-\sum_{k=0}^N y_k(t)(f(k)-1)-1\Bigg|\,&=\,
\Bigg|\sum_{k\geq N+1} y_k(t)(f(k)-1)\Bigg|\,<\,\d\,.
\end{align*}
We have, for every $k\geq0$ and $t\geq0$,
\begin{align*}
x_k(t)\,&=\,
x_k(0)+
\int_0^t \Bigg(
\sum_{i=0}^\ell 
x_i(s)f(i)M_H(i,k)-x_k(s)\phi_H(s)
\Bigg)ds\,,\\
y_k(t)\,&=\,
y_k(0)+
\int_0^t \Bigg(
\sum_{i=0}^k 
y_i(s)f(i)M_\infty(i,k)-y_k(s)\phi_\infty(s)
\Bigg)ds\,.
\end{align*}
Thus, for every $t\in[0,T]$,
\begin{multline*}
|x_k(t)-y_k(t)|\,\leq\,
|x_k(0)-y_k(0)|+\\
\sum_{i=0}^k
\int_0^t
\big|x_i(s)f(i)M_H(i,k)-y_i(s)f(i)M_\infty(i,k)\big|ds+\\
\sum_{i=k+1}^\ell 
\int_0^t
x_i(s)f(i)M_H(i,k)ds+
\int_0^t
|x_k(s)\phi_H(s)-y_k(s)\phi_\infty(s)|ds
\,.
\end{multline*}
The first term on the right is bounded by $\d$.
Adding and subtracting the quantity $x_i(s)f(i)M_\infty(i,k)$
in each of the terms in the first sum,
we see that the first sum is bounded by 
$$\sum_{i=0}^k f(i)M_\infty(i,k)\int_0^t
|x_i(s)-y_i(s)|ds+f(0)\d T\,.$$
The second sum is bounded by $f(0)\d T$,
and for the last term,
adding and subtracting $y_k(s)\phi_H(s)$ 
inside the integral, we have
\begin{multline*}
\int_0^t
|x_k(s)\phi_H(s)-y_k(s)\phi_\infty(s)|ds
\,\leq\\
\int_0^t|x_k(s)-y_k(s)|\phi_H(s)ds
+\int_0^t y_k(s)|\phi_H(s)-\phi_\infty(s)|ds\,.
\end{multline*}
Noting that $\phi_H(s)\leq f(0)$ 
and $y_k(s)\leq 1$ for all $s$,
we deduce
from the bounds on $\phi_H$ and $\phi_\infty$
that the above expression is bounded by
$$f(0)\int_0^t|x_k(s)-y_k(s)|ds+
2\d T+\sum_{i=0}^N |f(i)-1|\int_0^t|x_i(s)-y_i(s)|ds\,.
$$
Let $C$ be the maximum of the $|f(i)-1|$ for $1\leq i\leq N$,
it follows that
\begin{multline*}
|x_k(t)-y_k(t)|\,\leq\,
\d\big(
1+2(f(0)+1)T
\big)+\\
\sum_{i=0}^k f(0)M_\infty(i,k)\int_0^t
|x_i(s)-y_i(s)|ds+\\
f(0)\int_0^t|x_k(s)-y_k(s)|ds
+C\sum_{i=0}^N\int_0^t|x_i(s)-y_i(s)|ds\,.
\end{multline*}
We sum for $0\leq k\leq n$ and we get
\begin{multline*}
\sum_{k=0}^n |x_k(t)-y_k(t)|\,\leq\,
\d (n+1)\big(
1+2(f(0)+1)T
\big)+\\
\sum_{i=0}^n f(0)\Bigg(\sum_{k=i}^n M_\infty(i,k)\Bigg)
\int_0^t
|x_i(s)-y_i(s)|ds+\\
\sum_{k=0}^n f(0)\int_0^t|x_k(s)-y_k(s)|ds
+
(n+1)C\sum_{i=0}^n\int_0^t|x_i(s)-y_i(s)|ds\,.
\end{multline*}
We deduce that
$$\sum_{k=0}^n|x_k(t)-y_k(t)|\,\leq\,\d C_1
+C_2\int_0^t
\Bigg(\sum_{k=0}^n\big|x_k(s)-y_k(s)\big|
\Bigg)ds\,,$$
where $C_1,C_2$ positive constants that do not depend on 
$\ell$ or $q$.
We conclude thanks to Gronwall's lemma,
by choosing $\d<\e C_1^{-1}e^{-C_2 T}$.

\section{Convergence of the stationary solution}\label{Convstat}
%
%
Finally, 
we proceed to the proof of theorem~\ref{convpstat}.
Let us recall that the matrix
$(W(i,j),0\leq i,j\leq\ell)$
is defined by
$$\forall\, i,j\in\zl\,,\qquad
W(i,j)\,=\,f(i)M_H(i,j)\,.$$
The vector $x^*$ solves the equation
$$\phi_H x_k^*\,=\,
\sum_{i=0}^\ell x_i^* W(i,k)\,,\qquad
0\leq k\leq\ell\,,$$
where
$$\phi_H\,=\,
\sum_{i=0}^\ell x^*_i f(i)\,,$$
is also the Perron--Frobenius eigenvalue of $W$.
In particular,
$\phi_H\in\,]0,f(0)[\,$.
Up to the extraction of a subsequence,
we can suppose the existence of the limits
$$\phi_\infty\,=\,
\lim_{\genfrac{}{}{0pt}{1}{\ell\to\infty,\,q\to0}{\ell q\to a}}\,
\phi_H\,,\qquad
y^*_k\,=\,
\lim_{\genfrac{}{}{0pt}{1}{\ell\to\infty,\,q\to0}{\ell q\to a}}\,
x^*_k\,,\quad k\geq 0\,.$$
%
%
Writing down the $k$-th equation of the system $(x^*)^T\phi_H =(x^*)^T W$,
we conclude that
$$
\sum_{i=0}^k x^*_i f(i) M_H(i,k)
\,<\,\phi_H x^*_k\,<\,\\
\sum_{i=0}^k x^*_i f(i) M_H(i,k)
+f(0)\max_{k<i\leq\ell}M_H(i,k)\,.
$$
In particular, 
if we take the left inequality with $k=0$, 
and if we divide both sides by $x^*_0$,
we get, passing to the limit,
that $\phi_\infty\geq f(0)\exa$.
Passing to the limit in the above inequalities,
we obtain the system of equations
$$\phi_\infty y^*_k\,=\,
\sum_{i=0}^k y^*_i f(i) \exa\frac{a^{k-i}}{(k-i)!}\,,\qquad
k\geq0\,.$$
The zeroth equation reads $\phi_\infty y_0^*=y_0^*f(0)\exa$.
Since the sum of the components of the 
vector $(x^*_k)_{k\geq0}$ 
is equal to 1,
the sequence $(y_k^*)_{k\geq0}$
satisfies
$$\sum_{k\geq0}y_k^*\,\leq\,1\,.$$
We know from~\cite{CDClassdep}
that this system of equations
only admits the solution $y_k^*=0$ when $f(0)\exa<1$.
On the other hand, if $f(0)\exa>1$, 
we see that necessarily $y_0^*>0$: 
indeed,
if $K$ is the first index $k\geq 0$ 
such that $y_k^*>0$,
it follows from a passage to the limit in
the above inequalities that
$\phi_\infty\leq f(K)\exa$.
In view of the constraints 
$\phi_\infty\geq f(0)\exa$
and $f(0)>f(k)$ for all $k\geq 1$, we 
deduce that $K$ must be equal to $0$.
Likewise, if $y_k^*=0$ for every $k\geq0$,
then taking $N$ large enough so that,
for all $n\geq N$, $|f(n)-1|<\e$,
it follows that
$$\phi_H\,=\,
\sum_{k=0}^\ell x^*_k f(k)
\,\leq\,\sum_{k=0}^N x^*_k f(k)
+(1+\e)\,.$$
We deduce from here that $\phi_\infty\leq 1+\e$,
which, for $\e$ small enough is in 
contradiction with the fact that $\phi_\infty\geq f(0)\exa$.
Therefore, $y_0^*>0$ and $\phi_\infty=f(0)\exa$.
Summing over $k\geq0$ 
in the above system of equations,
we see that
$$\phi_\infty\sum_{k\geq0}y_k^*\,=\,
\sum_{i\geq0}y_i^*f(i)\,=\,\phi_\infty\,.$$
We conclude that the components of 
$(y_k^*)_{k\geq0}$ add up to 1,
thus, $(y_k^*)_{k\geq0}$
must be equal to $\rho^0$.
%
\section*{Acknowledgements}
This work was supported by a public grant as part of the
Investissement d'avenir project, reference ANR-11-LABX-0056-LMH,
LabEx LMH. The author wishes to thank Michel Bena\"im and Rapha\"el Cerf 
for the valuables discussions that contributed to improve the article.

\bibliography{qs}
\bibliographystyle{plain}

\appendix

\section{Lemmas on linear ODEs}
\label{LemmesEdo}
We give here some lemmas concerning linear ODEs,
and specially their long time behavior.
\begin{lemma}\label{posilemma}
Let $\a:[0,+\infty[\,\lra[0,+\infty[\,$
and $\b:[0,+\infty[\,\lra\R$
be Lipschitz functions and let $(z(t),t\geq0)$
be the solution of the differential equation 
$$z'(t)\,=\,\a(t)+\b(t)z(t)\,.$$
If $z(0)\geq0$ then $z(t)\geq0$ for all $t\geq 0$.
\end{lemma}
\begin{proof}
The trajectory $(z(t),t\geq0)$ is continuous.
If there exists $t^*\geq0$ such that $z(t^*)=0$, then
$$z'(t^*)\,=\,\a(t^*)\,\geq\,0\,,$$
and thus $z(t)\geq 0$ for all $t\geq0$.
\end{proof}
\begin{lemma}\label{compedolin}
Let $\a,\widetilde{\a}:[0,+\infty[\,\lra[0,+\infty[\,$
and $\b,\widetilde{\b}:[0,+\infty[\,\lra\R$ 
be Lipschitz functions satisfying
$$\forall\,t\geq0\,,\qquad
\a(t)\,\leq\,\widetilde{\a}(t)\,,\qquad
\b(t)\,\leq\,\widetilde{\b}(t)\,.$$
Let $(y(t),t\geq0)$ and $(z(t),t\geq0)$
be the solutions of the ODEs
$$y'(t)\,=\,\a(t)+\b(t)y(t)\,,\qquad
z'(t)\,=\,\widetilde{\a}(t)+\widetilde{\b}(t)z(t)\,.$$
If $z(0)\geq y(0)\geq 0$ then $z(t)\geq y(t)$ 
for all $t\geq 0$.
\end{lemma}
\begin{proof}
We have
$$z'(t)-y'(t)\,=\,
\widetilde{\a}(t)-\a(t)+(\widetilde{\b}(t)-\b(t))z(t)
+\widetilde{\b}(t)(z(t)-y(t))\,.$$
From the previous lemma, $z(t)\geq0$ for all $t\geq 0$.
Thus, applying the previous lemma once again,
$z(t)-y(t)\geq0$ for all $t\geq0$.
\end{proof}
\begin{lemma}
Let $\a,\b:[0,+\infty[\,\lra[0,+\infty[\,$ 
be Lipschitz functions, and suppose that there exist
$\a^*,\b^*\in\,]0,+\infty[\,$ such that
$$\lim_{t\to\infty}\a(t)\,=\,\a^*\,,\qquad
\lim_{t\to\infty}\b(t)\,=\,\b^*\,.$$
Let $(y(t),t\geq0)$
be the solution of the differential equation
$$y'(t)\,=\,\a(t)-\b(t)y(t)\,.$$
Then, for every initial condition $y(0)\in\R$,
$$\lim_{t\to\infty}y(t)\,=\,\frac{\a^*}{\b^*}\,.$$
\end{lemma}
\begin{proof}
Let $\e>0$ be small enough so that $\a^*-\e,\b^*-\e>0$.
Let $T\geq0$ be large enough so that
$$\forall\,t\geq T\,,\qquad
|\a(t)-\a^*|\,<\,\e\,,\qquad
|\b(t)-\b^*|\,<\,\e\,.$$
Let $(\underline{y}(t),t\geq0)$ and
$(\overline{y}(t),t\geq0)$ 
be the solutions of the differential equations
$$\underline{y}'(t)\,=\,(\a^*-\e)-(\b^*+\e)\underline{y}'(t)\,,\qquad
\overline{y}'(t)\,=\,(\a^*+\e)-(\b^*-\e)\overline{y}'(t)\,,$$
with $\underline{y}(0)=
\overline{y}(0)=y(T)$.
From the previous lemma,
for all $t\geq0$,
$$\underline{y}(t)\,\leq\,
y(T+t)\,\leq\,
\overline{y}(t)\,.$$
Yet, $\underline{y}(t)$ and $\overline{y}(t)$ converge:
$$\lim_{t\to\infty}\underline{y}(t)\,=\,
\frac{\a^*-\e}{\b^*+\e}\,,\qquad
\lim_{t\to\infty}\overline{y}(t)\,=\,
\frac{\a^*+\e}{\b^*-\e}\,.$$
We conclude that
$$\frac{\a^*-\e}{\b^*+\e}\,\leq\,
\liminf_{t\to\infty} y(t)\,\leq\,
\limsup_{t\to\infty}y(t)\,\leq\,
\frac{\a^*+\e}{\b^*-\e}\,.$$
We send $\e$ to 0 and we obtain the desired result.
\end{proof}
\end{document}